\documentclass{amsproc}

\usepackage{graphicx}
\usepackage{color}
\usepackage{caption}
\usepackage{comment}
\usepackage{subcaption}
\usepackage[dvipsnames]{xcolor} 

\usepackage[section]{algorithm}
\usepackage{}

\newtheorem{theorem}{Theorem}[section]

\theoremstyle{definition}
\newtheorem{definition}[theorem]{Definition}

\newtheorem{ques}[theorem]{Question}
\newtheorem{pr}[theorem]{Problem}
\newtheorem{prop}[theorem]{Proposition}

\theoremstyle{remark}

\numberwithin{equation}{section}

\DeclareMathOperator{\Span}{Span}

\DeclareMathOperator{\Wing}{Wing}
\RequirePackage[normalem]{ulem} 
\RequirePackage{color}\definecolor{RED}{rgb}{1,0,0}\definecolor{BLUE}{rgb}{0,0,1} 

\begin{document}

\title[Modeling the distribution of distance data in Euclidean space]{Modeling the distribution of \\ distance data in Euclidean space}


\author{Ruth Davidson}
\address{Department of Mathematics, University of Illinois Urbana-Champaign}
\curraddr{}
\email{redavid2@illinois.edu}
\thanks{ }

\author{Joseph Rusinko}
\address{Department of Mathematics and Computer Science, Hobart and William Smith Colleges}
\curraddr{}
\email{rusinko@hws.edu}
\thanks{}

\author{Zoe Vernon}
\address{Department of Mathematics, Washington University in St. Louis}
\curraddr{}
\email{zoe.vernon@wustl.edu}
\thanks{}

\author{Jing Xi}
\address{Department of Mathematics, North Carolina State University}
\curraddr{}
\email{jxi2@ncsu.edu}
\thanks{}


\date{}

\begin{abstract}
Phylogenetic inference-the derivation of a hypothesis for the common evolutionary history of a group of species- is an active area of research at the intersection of biology, computer science, mathematics, and statistics. One assumes the data contains a phylogenetic signal that will be recovered with varying accuracy due to the quality of the method used, and the quality of the data. 

The input for distance-based inference methods is an element of a Euclidean space with coordinates indexed by the pairs of organisms. For several algorithms there exists a  subdivision of this space into polyhedral cones such that inputs in the same cone return the same tree topology.  The geometry of these cones has been used to analyze the inference algorithms.  In this chapter, we model how input data points drawn from DNA sequences are distributed throughout Euclidean space in relation to the space of tree metrics, which in turn can also be described as a collection of polyhedral cones.

\end{abstract}

\maketitle





\section{Introduction to Phylogenetics}\label{sec:Intro}

A \emph{phylogeny} is a mathematical model of the common evolutionary history of a group of \emph{taxa} $X$, where $X$ may be a set of genes, species, or multiple individuals sampled within a population. Phylogenies are commonly represented by a \emph{phylogenetic tree} $T$ which is a connected, acyclic graph in which the degree-one vertices, or \emph{leaves} of $T$, are labeled with the elements of $X$. The tree $T$ is a representation of the evolutionary history of the set $X$ \cite{Semple_Steel:2003}. 

In many phylogenetic reconstruction pipelines one creates a nucleotide or protein sequence alignment, uses the alignment to estimate pairwise dissimilarities between taxa, and then uses the set of dissimilarities to infer a phylogenetic tree.  One class of methods used in the third step of such a pipeline are distance-based algorithms: examples include Neighbor Joining (NJ) \cite{Saitou_Nei:1987}, Unweighted Pair Group Method with Arithmetic Mean (UPGMA) \cite{Sokal_Michener:1958, Sokal_Sneath:1963} (Algorithm \ref{UPGMA}), and FastME \cite{FastME}.  These algorithms seek an edge-weighted tree such that the pairwise distances between taxa on the tree reflect the dissimilarities computed from the data. While not as popular as maximum-likelihood methods, distance-based algorithms remain relevant for reasons beyond their speed relative to other methods.  For example, they serve as components of other popular methods, such as the species tree inference method NJst \cite{NJst} and the DNA sequence alignment method MUSCLE \protect\cite{edgar2004muscle}.

It is possible to analyze the theoretical properties of distance-based algorithms by considering which sets of input data return particular output trees. The set of input data can be viewed as a subset of Euclidean space. It was shown geometrically that NJ solves its optimization problem well in low dimension in \cite{Eickmeyer_Huggins_Pachter_Yoshida:2008}, and in \cite{Davidson_Sullivant:2013, Davidson_Sullivant:2014}, geometry provided new insights into biases in the structure of trees reconstructed from phylogenetic data \cite{Aldous}.  In this chapter we use computational tools in a geometric setting to observe that pairwise dissimilarity data computed from real sequence data will not be uniformly distributed across Euclidean space, but instead will be clustered around an object corresponding to pairwise distances between taxa induced by phylogenetic trees.  We derive models for the distribution of this data and discuss possible applications for these models and tools.  

\section{The Space of additive distance matrices embedded in $\mathbb{R}^{n \choose 2}$}\label{additive} 
\subsection{Dissimilarity maps used in distance-based phylogenetic inference as points in Euclidean space}

\begin{definition}\label{diss}
For a fixed set $X$ of $n$ taxa, a \emph{dissimilarity map} is an $n \times n$ symmetric matrix $\delta$ 
where $\delta(u,v)$ is a measure of evolutionary difference between taxa $u$ and $v$. 

\end{definition}

There are many popular ways to estimate such a measure that assume models of sequence evolution such as Jukes-Cantor (JC) \cite{Jukes_Cantor:1969}, Kimura two-parameter (K2P) \cite{Kimura:1980}, or the general time-reversible model (GTR) \cite{GTR}. For example, the JC model of DNA sequence evolution corresponds to a measure of dissimilarity between two DNA sequences of the same length representing taxa $u$ and $v$: 

\begin{equation}\label{JCCorrection}
\delta(u,v) = -\frac{3}{4}\log\left(1-\frac{4}{3} \left(\frac{h}{l}\right)\right)
\end{equation}

where $h$ is the number of places that the sequences differ and $l$ is the shared length of the two sequences. Equation \ref{JCCorrection} is known as the \emph{Jukes-Cantor Correction} because it returns a measure of dissimilarity between taxa based on the expected number of mutations given by the JC model of sequence evolution instead of the observed number of differences in the sequences.

Dissimilarity maps are the inputs to distance-based algorithms. By our definition,  $\delta(u,v) = \delta(v,u)$ for all pairs $\{u, v\} \subset X$.  Furthermore, $\delta(u,u) = 0$ because $\delta$ is a measure of evolutionary distance. So each dissimilarity map $\delta$ is also a vector in the Euclidean space $\mathbb{R}^{n \choose 2}$ with coordinates labeled by all distinct pairs of taxa. Finally, we can assume $\delta$ is in the positive orthant, which we denote $\mathbb{R}^{n \choose 2}_{\geq 0}$; as we see in Equation \ref{JCCorrection}, popular dissimilarity measures such as JC, K2P, and GTR return positive numbers as estimates of evolutionary difference. 

The output of a distance-based algorithm is an additive distance matrix (sometimes called a tree metric, though this term has more than one usage in the literature).  An \emph{additive distance matrix}, which we will denote as $d(u,v)$ or simply $d$, can be realized by a phylogenetic tree $T$ with edge weights that are real numbers, where $d(u,v)$ for two distinct elements $u$ and $v$ of $X$ is the sum of the weights of the edges in the unique path between $u$ and $v$ in $T$.  \emph{Ultrametric} additive distance matrices can be realized by a tree with a root vertex $r$ that is not labeled with a element of $X$, and satisfy the property that for every two elements $u$ and $v$ of $X$, $d(u,r) = d(v,r)$. Figure \ref{fig:ultrametric} shows an ultrametric on three taxa.

\begin{figure}[ht!]
\centering

\includegraphics[width=4cm]{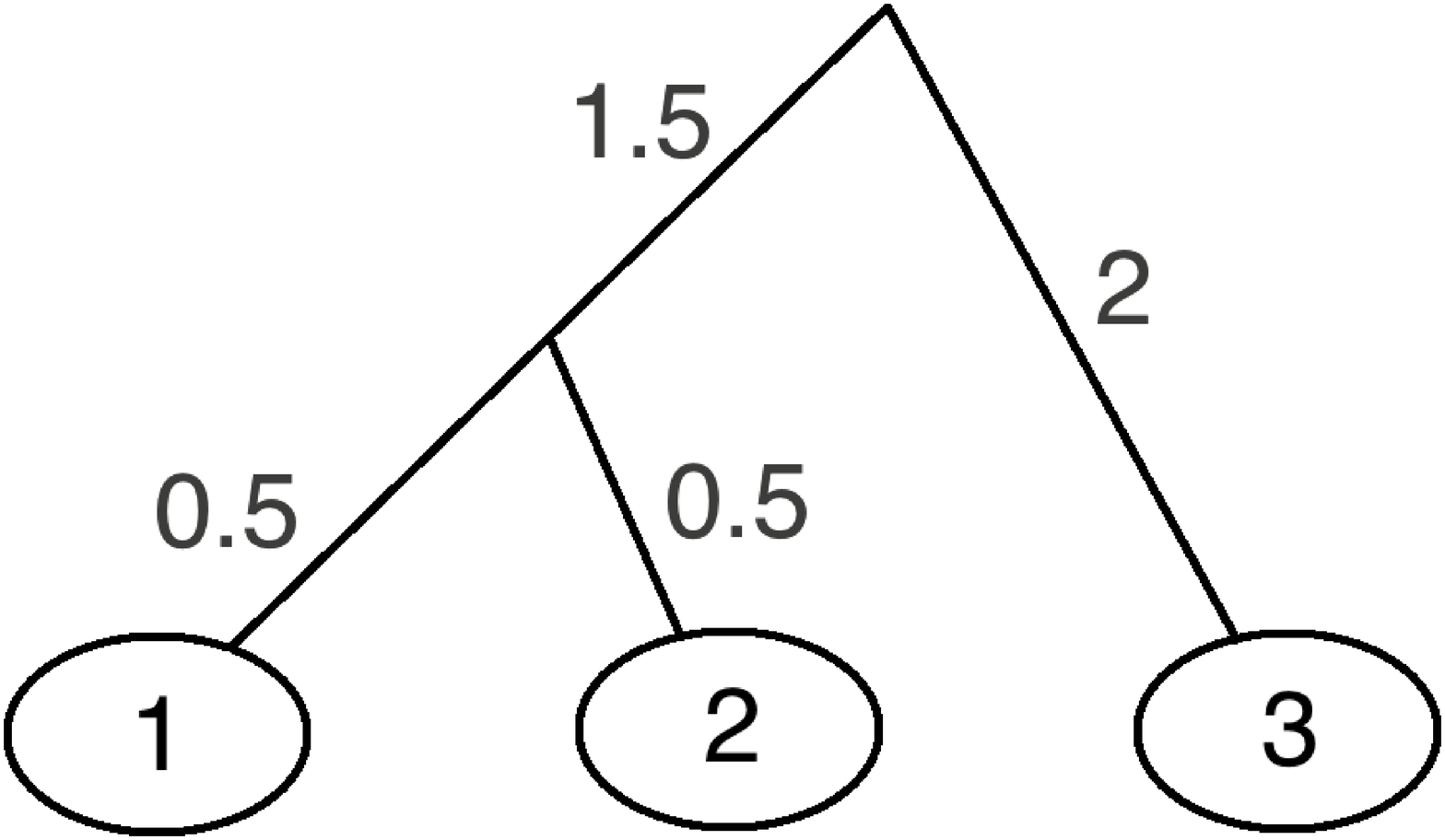}

\caption{An ultrametric on three taxa}
\label{fig:ultrametric}
\end{figure}

\subsection{Polyhedral descriptions of the outputs of distance-based methods as points in Euclidean space}

Recall that a \emph{polyhedral cone} is the nonnegative span of a set of vectors in $\mathbb{R}^{d}$ for $d \geq 1$, commonly referred to as the $\mathcal{V}$-representation.  Since polyhedral cones are polyhedra, they also have an $\mathcal{H}$-representation as the closed (in the traditional sense of the usual topology on $\mathbb{R}^{d}$) intersection of a finite number of half-spaces defined by hyperplanes. A \emph{face} of a polyhedral cone is the intersection of the cone with a hyperplane that defines a half-space in $\mathbb{R}^{d}$ that entirely contains the cone.  See \cite{ziegler} for a comprehensive introduction to polyhedral theory.

\begin{definition}
\label{fan_defn}
  A \emph{fan} is a family $\mathcal{F}$  of polyhedral cones in $\mathbb{R}^{d}$ such that:
        \begin{enumerate}
        \item if $P \in \mathcal{F}$ then every nonempty face of $P$ is in $\mathcal{F}$, and 
        \item if $P_{1}, P_{2} \in \mathcal{F}$ then $P_{1} \cap P_{2}$ is a face of both $P_{1}$ and $P_{2}$.
        \end{enumerate}
        
\end{definition}

Informally, polyhedral fans are special collections of cones that are easy to work with because they ``play well together." The set of all additive distance matrices on $n$ taxa,  denoted $\mathcal{T}_{n}$, and the set of ultrametric distance matrices on $n$ taxa, denoted $\mathcal{ET}_{n}$, are each classified as both a polyhedral fan and a tropical variety. These spaces are studied in \cite{Ardila_Klivans:2006} and \cite{Speyer_Sturmfels:2004}, respectively. We follow the notation of \cite{Davidson_Sullivant:2014} to describe these spaces and note that $\mathcal{ET}_{n} \subset \mathcal{T}_{n} \subset \mathbb{R}^{n \choose 2}$ for all $n \geq 3$.  In other words, the spaces of additive matrices  $\mathcal{ET}_{n}$ and $\mathcal{T}_{n}$ are \emph{embedded} in the Euclidean space of dissimilarity maps in a very natural way.  

In this chapter we model how three-taxon samples of distance data  drawn from dissimilarity maps on $n$ taxa are distributed in relation to $\mathcal{ET}_{3}$. We model this data in relation to $\mathcal{ET}_{3}$ rather than $\mathcal{T}_{3}$ because (1) most of the restrictions of the samples to three taxa already lie in $\mathcal{T}_{3}$ and (2) the condition of ultrametricity is of particular interest to biologists as it indicates that the molecular clock assumption is valid on the sample, which means that evolution occurred at a constant rate over time throughout the entire phylogeny.

We write the restriction of a dissimilarity map to a three-taxon sample $\delta$ as a point $(x,y,z) = (\delta(u,v), \delta(u,w), \delta(v,w))$ in $\mathbb{R}^{3 \choose 2} = \mathbb{R}^{3}$. We wish to understand geometric properties of the distribution of points $\delta$ in relation to $\mathcal{ET}_{3}$.

The space $\mathcal{ET}_{3}$ (see Figure \ref{fig:ET3}) is a two-dimensional polyhedral fan consisting of three two-dimensional cones
$$
\Span_{\ge 0} \{ (1,1,1), (0,1,1) \}, \  \Span_{\ge 0} \{ (1,1,1), (1,0,1) \}, $$ and
$$
\Span_{\ge 0} \{ (1,1,1), (1,1,0) \},
$$
whose pairwise intersections consist solely of the one-dimensional cone $\Span_{\ge 0} \{ (1,1,1)\}$. We refer to the two-dimensional (top-dimensional) cones in the fan $\mathcal{ET}_{3}$ as the \emph{wings}, where $$\Wing \ 1 = \Span_{\geq 0} \{ (1,1,1), (0,1,1) \},$$ $$\Wing \ 2 = \Span_{\geq 0} \{ (1,1,1), (1,0,1) \},$$ and $$\Wing \ 3 = \Span_{\geq 0} \{ (1,1,1), (1,1,0) \}.$$ We refer to the one-dimensional cone in $\mathcal{ET}_{3}$ $\Span_{\ge 0} \{ (1,1,1)\}$ as the \emph{spindle}. 

\begin{figure}[h!]
\centering

\includegraphics[width=0.3\textwidth]{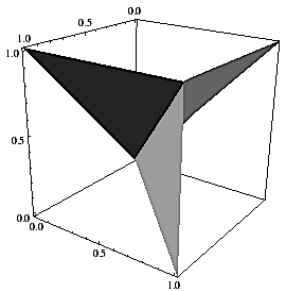}

\caption{Ultrametric outputs $\mathcal{ET}_{3}$}
\label{fig:ET3}
\end{figure}

We introduce a new set of coordinates which describe a dissimilarity map in relation to the geometry of $\mathcal{ET}_{3}$.
Let $\delta$ be a dissimilarity map. The relationship of $\delta$ to $\mathcal{ET}_{3}$ can be described by (1) $W$, the nearest wing to $\delta$, (2) $D_{1}$, the signed distance from $\delta$ to the nearest wing, (3) $D_{2}$, the distance from the point on the nearest wing to the spindle, and (4) $D_{3}$, the distance from the point on the spindle to the origin $(0,0,0)$. 
Then $\delta$ can be written as a sum of three pairwise-orthogonal vectors with lengths $D_{1}, \ D_{2}$, and $D_{3} $ representing a path from $\delta$ to the origin that begins with a choice of the nearest wing. We call the vector $(W, \ D_{1}, \ D_{2}, \ D_{3} )$ the \emph{path trace} of $\delta$, where $W = \Wing \ i , \ i \in \{1, 2, 3 \}$. The path trace is an equivalent representation of $\delta$ as $\delta$ can be uniquely determined by the path trace.

When analyzing a distance-based algorithm, one should emphasize dissimilarity maps likely to arise as outputs of a distance-approximation method based on a popular statistical model of sequence evolution such as JC, K2P, and GTR, which were introduced in Section \ref{sec:Intro}. In the following section we will review distance-based methods and see that they also induce polyhedral structures; studying these structures has shed light on the behavior of the algorithms themselves.

\section{Geometry of distance-based methods}
\subsection{Optimization problems associated to distance-based methods}
\begin{pr}\label{LSP}The Least-Squares Phylogeny problem (LSP) asks, for a given dissimilarity map $\delta \in \mathbb{R}^{n \choose 2}$, 
what is the additive distance matrix $d$ that minimizes the ordinary Euclidean distance given by the formula
$$
 \sqrt{\sum_{x,y \in X} 
(\delta(x,y) - d(x,y))^{2} }. 
$$

\end{pr}

The distance-based methods Unweighted Pair Group Method with Arithmetic Mean (UPGMA)  \cite{Sokal_Michener:1958, Sokal_Sneath:1963} and Neighbor-Joining (NJ) \cite{Saitou_Nei:1987} are approximations to LSP \cite{Hosten,Saitou_Nei:1987}.  However, an alternative interpretation of NJ \cite{GascuelSteel, DesperGascuel} is that NJ performs a heuristic search, guided by a linear transformation of $\delta$ known as the \emph{$Q$-criterion} at each agglomeration step, that minimizes a tree-length estimate due to Yves Pauplin \cite{Pauplin} known as the ``Balanced Minimum Evolution" (BME) criterion.

Day showed that both Problem \ref{LSP} and finding the tree minimizing the BME criterion are NP-hard \cite{Day1987}.  Thus polynomial-time, heuristic, distance-based algorithms such as BIONJ \cite{BIONJ}, Weighbor \cite{Weighbor}, and FastME \cite{FastME} remain essential for solving Problem \ref{LSP} as well as for finding a tree minimizing the BME criterion. These algorithms take dissimilarity maps as inputs,
may outperform NJ in terms of topological accuracy under certain conditions, and exhibit superior immunity to reconstruction pathologies well-known to biologists such as long-branch attraction \cite{LongBranch}. 

 \subsection{Polyhedral decompositions of $\mathbb{R}^{n \choose 2}$ induced by distance-based algorithms}
 
Given a distance-based algorithm, there exists a closed-form description of the input space for a phylogenetic reconstruction method due to the natural subdivision of said space induced by the decision criteria of the method. This subdivision consists of regions containing dissimilarity maps that return additive distance matrices realized by the same combinatorial type of tree upon application of the method. For example, the distance-based algorithms NJ and UPGMA induce subdivisions of the input space into families of polyhedral cones defined by the sets of linear inequalities defining the decision steps in these algorithms. To illustrate, we review why there is a polyhedral description of the UPGMA algorithm here.

\begin{algorithm}[h!]
\caption{UPGMA}
\label{UPGMA}

\begin{itemize}

\item Input: a dissimilarity map $\delta \in \mathbb{R}^{ {n(n-1) / 2}}_{\geq 0}$ on $[n]$.
\item Output: an ultrametric $d \in \mathbb{R}^{ {n(n-1) / 2}}_{\geq 0}$.

\begin{enumerate}
\item Initialize an unordered set partition $\pi_{n} = 1|2| \cdots | n$, and set $\delta^{n} = \delta$.
\item For $i = n-1, \ldots, 1$ do
\begin{itemize}
\item From partition $\pi_{i+1} = \lambda^{i+1}_{1} | \cdots | \lambda^{i+1}_{i+1}$
and distance vector $\delta^{i+1} \in \mathbb{R}^{(i+1)i/2}_{\geq 0}$
choose $j, k$ be so that $\delta^{i+1}(\lambda^{i+1}_{j}, \lambda^{i+1}_{k})$ is minimized.
\item Set $\pi_{i}$ to be the partition obtained from $\pi_{i+1}$ by
 merging $\lambda^{i+1}_{j}$ and $ \lambda^{i+1}_{k}$ and leaving all other
 parts the same. Let $\lambda^{i}_{i} = \lambda^{i+1}_{j} \cup \lambda^{i+1}_{k}$.
\item Create new distance $\delta^{i} \in \mathbb{R}^{i(i-1)/2}_{\geq 0}$ by
$\delta^{i}(\lambda, \lambda') = \delta^{i+1}(\lambda, \lambda')$ if $\lambda, \lambda'$ are
both parts of $\pi_{i+1}$ and
$$
\delta^{i}(\lambda, \lambda^{i}_{i})
= \frac{ |\lambda^{i+1}_{j}|}{ |\lambda^{i}_{i}|} 
\delta^{i+1}( \lambda, \lambda^{i+1}_{j} ) + 
\frac{|\lambda^{i+1}_{k}|}{ |\lambda^{i}_{i}|} 
\delta^{i+1}( \lambda, \lambda^{i+1}_{k} )
$$
otherwise.
\item For each $u \in \lambda^{i+1}_{j}$ and $v \in\lambda^{i+1}_{k}$, 
set $d(u,v) = \delta^{i+1}(\lambda^{i+1}_{j}, \lambda^{i+1}_{k})$.
\end{itemize}
\item Return: A combinatorial rooted tree $T$ with edge weights $w : E(T) \rightarrow \mathbb{R}_{\geq 0}$ and
the ultrametric $d_{T, w}$ realized by $T$.

\end{enumerate}

\end{itemize}

\end{algorithm}

The inequalities coming from the decision steps in UPGMA (Algorithm \ref{UPGMA}) are linear combinations of the original input coordinates, and define a family of polyhedral cones that completely partition the input space. This representation of a polyhedral cone in terms of a set of inequalities is known as an \emph{$H$-representation}. The \emph{$V$-representation} of these cones-i.e. a description in terms of their extreme rays-is given in \cite{Davidson_Sullivant:2013}. Knowledge of both the $H$- and $V$-representations of the cones into which UPGMA divides the input space gives a complete discrete geometric description of the input space, allowing for computations and further conjectures about how this discrete geometric description might affect the performance and behaviors of the algorithm in practice. 

While UPGMA outputs ultrametrics, and so maps $\mathbb{R}^{n \choose 2}_{\geq 0}$ to cones in $\mathcal{ET}_{n}$ indexed by combinatorial rooted trees, Neighbor Joining maps $\mathbb{R}^{n \choose 2}_{\geq 0}$ to cones in $\mathcal{T}_{n}$ indexed by unrooted tree shapes. As is the case with UPGMA, NJ partitions the input space into a family of polyhedral cones with an $H$-representation arising naturally from the decision criterion in Step (2) of Algorithm \ref{UPGMA} which are linear inequalities in terms of the coordinates of the input space. Though the problem has been studied in \cite{NJRays}, a complete description of the $V$-representation, or extreme rays, for the cones in the polyhedral partition of the input space induced by NJ is still unknown.

\subsection{Using polyhedral geometry to analyze distance-based algorithms}
The intersection of the polyhedral decomposition of the input space cones with the piece of the hypersphere $\mathbb{S}_{{n \choose 2}-1}$ in the section of a Euclidean space with all non-negative coordinates, which we denoted as $HS_{\ge 0}$ in Section \ref{additive}, provides a finite measure of how frequently the algorithm returns a particular topology sometimes called the \emph{spherical volume} of the cone. Estimation of spherical volumes of polyhedral cones has been used to study the behavior of distance-based phylogenetic methods on small numbers of taxa. For example, in \cite{Eickmeyer_Huggins_Pachter_Yoshida:2008} estimated spherical volumes were used to assess the agreement of regions in the input space for NJ and the Balanced Minimum Evolution (BME) criterion, which as mentioned above is one of the optimization problems NJ seeks to solve. Also, in \cite{Davidson_Sullivant:2013}, estimated spherical volumes of cones in the input space for UPGMA indicated that UPGMA may be biased against unbalanced tree topologies, and in \cite{Davidson_Sullivant:2014} the notion of spherical volume of small neighborhoods of the input space was further used to investigate the behavior of NJ and UPGMA. These findings assume that observed dissimilarity maps are uniformly distributed when projected on the sphere. 

\subsection{Limitation to the uniform distribution model}
A key biological motivation for properly modeling the distribution of distance data is the so-called \emph{rogue taxa phenomenon} in which the inclusion of a taxon in the estimation of a phylogeny results in reduced accuracy. See \cite{PruningRogues} for an example of a recent method designed to deal with this problem. In \cite{Cueto_Matsen:2011} polyhedral geometry was used to investigate the impact of including an additional taxon on the topology of a BME tree. In a simulation the authors showed that by including distances to an extra taxon and therefore lifting the problem to a higher-dimensional space, one could completely transform the BME topology for the set of taxa corresponding to the original distance matrix in a large number of cases, and that the effect worsened as the number of taxa grew. 

But in \cite{Westover_Rusinko_Hoin_Neal:2013} the effect of adding one additional taxon to a small tree was studied for a viral data set as a biological analog of the simulation study done in \cite{Cueto_Matsen:2011}, which indicates that this effect may have been overstated due to the model of distances used in the simulation in \cite{Cueto_Matsen:2011}. Additional study of this effect using biological data may lead to more speculation about the true frequency of this effect. Motivated by the failure of the uniform distribution to capture the observed biological features we introduce two families of distribution functions in Section~\ref{sec:twomodels} that may better model the dissimilarity maps computed from DNA sequence data.

\section{Modeling dissimilarity maps using \\ statistical distributions and geometry}
\label{sec:twomodels}

\subsection{Two model families for dissimilarity maps}

We investigate two families of models of dissimilarity maps based on the underlying structure of $\mathcal{ET}_{n}$.  To construct our first model we use biological data as input, and fit a distribution function that describes  the individual coordinates of the path trace associated to the dissimilarity map. We refer to this viewpoint as Model Family (1). In the second model of distance data we assume the data has been generated under the Yule-Harding model with noise accounting for the non-ultrametric tree-like features in the data. Precisely, we use a normal random variable to account for the regular Euclidean distance of an input point from the space $\mathcal{ET}_{n}$.  We then compute the induced distribution of the associated coordinates of the path trace. We refer to this as Model Family (2).

\subsection{Using biological data to develop the two model families}

To ensure the comparison of these two models reflects biological data we fit the models to data drawn from the biological literature. TreeBASE \cite{Treebase} is a public, open-access website maintained by the Phyloinformatics Research Foundation, Inc. Data publicly available from TreeBASE is restricted to data associated to publications that have been submitted for peer review.  TreeBASE includes gene, species, and population phylogenies as well as the data used to infer these phylogenies. We used data matrices representing sequence alignments of nucleotide data, inferred distances from these data matrices using the open-source software package MEGA 5 \cite{MEGA}, and then extracted triples of distances from the large distance matrices to create large samples of dissimilarity maps in three dimensions.

In this chapter we highlight the findings for two data sets we obtained from TreeBASE. In \cite{Nuhn_Binder_Taylor_Halling_Hibbett:2013} the authors are testing support for the classification of Boletineae as a sub-order as well as the families Boletaceae and Paxillaceae, and they find support for these using a phylogenetic analysis of three different genes. The paper \cite{Spooner_Rojas_Bonierbale_Mueller_Srivastav_Senalik_Simon:2013} 
uses 9 genes to study the phylogenetic relationships of 29 species all within the same family: 22 of the taxa are in the \emph{Daucus} genus and 7 are from related genera. All 29 species belong to the Umbelliferae family.

\section{Developing Model Family (1)}
\label{sec:model1}
It is equally likely that a data point $\delta$ is closest to any of the three Wings, so $p_W(\Wing \ i ) = \Pr (W= \Wing \ i )= 1/3$, $i \in \{1, 2, 3 \}$. We model $D_{1}$ using a truncated normal distribution, and $D_{2}$ and $D_{3}$ are modeled in Model (1)-(A) using truncated generalized extreme value distributions (EVDs). For a comprehensive introduction to extreme value theory, see \cite{Coles:2001}. 
For Model (1)-(A), we designate the probability density function $f_{D_{3}}$ for $D_{3}$ to be the truncated EVD with parameters $\alpha_s$ and $\beta_s$ on domain $(0, +\infty)$ (the truncation forbids negative values):
$$
f_{D_3} (d_3|\alpha_s, \beta_s) = \frac{e^{\frac{\alpha_s-d_3}{\beta_s}-e^{\frac{\alpha_s-d_3}{\beta_s}}}}{\beta_s c_{\alpha_s,\beta_s}} I_{\{x > 0\}},
$$
where $c_{\alpha_s,\beta_s}$ is the normalizing constant and $I$ is the indicator function. 

Next, $f_{D_2|D_3} (d_2|d_3,\theta)$ follows the resulting distribution obtained by assuming that $r_1=\frac{d_2}{d_3}$ follows an EVD with $\alpha$ and $\beta$ on domain $(0, 1]$, where the truncation is determined by the geometric setting. Note that, conditional on $D_3=d_3$, $f_{D_2|D_3}$ also follows a truncated EVD. We consider $r_1$ first and let $d_2 = r_1 d_3$ so that the estimation of $\alpha$ and $\beta$ does not depend on the actual value of $d_3$. Similarly we define $r_2 = \frac{d_1}{\min(d_3+d_2/\sqrt{2}, \ d_2/2)}$ and model $r_{2}$ by the truncated normal distribution with mean 0 and variance $\sigma^2$ on domain $[-\frac{1}{\sqrt{3}}, \frac{1}{\sqrt{3}}]$. Then $d_1 = r_2 \min(d_3+d_2/\sqrt{2}, \ d_2/2)$ will also follow a truncated normal distribution with independently estimated variance $\sigma^2$.

The EVD distributions in Model (1)-(A) are replaced with Gamma distributions in Model (1)-(B). We chose the distributions for Model (1)-(A) and Model (1)-(B) for our model by fitting the path trace to histograms taken from triples of distances estimated from biological data using MEGA \cite{MEGA}.

\begin{figure}[ht]

 \begin{subfigure}{0.5\textwidth}
 \includegraphics[width=0.9\linewidth]{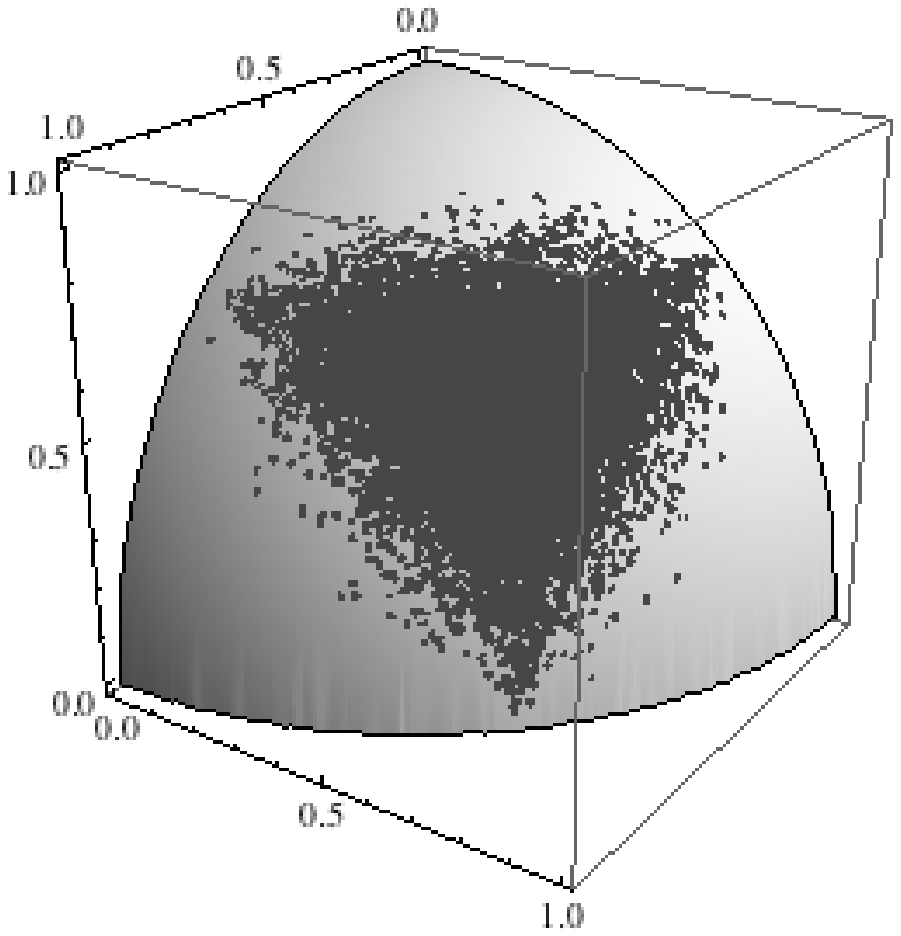}

 \caption{Data points from \cite{Nuhn_Binder_Taylor_Halling_Hibbett:2013}}
\label{fig:BolReal}
 \end{subfigure}
 ~ 
 \begin{subfigure}{0.5\textwidth}
 \includegraphics[width=0.9\linewidth]{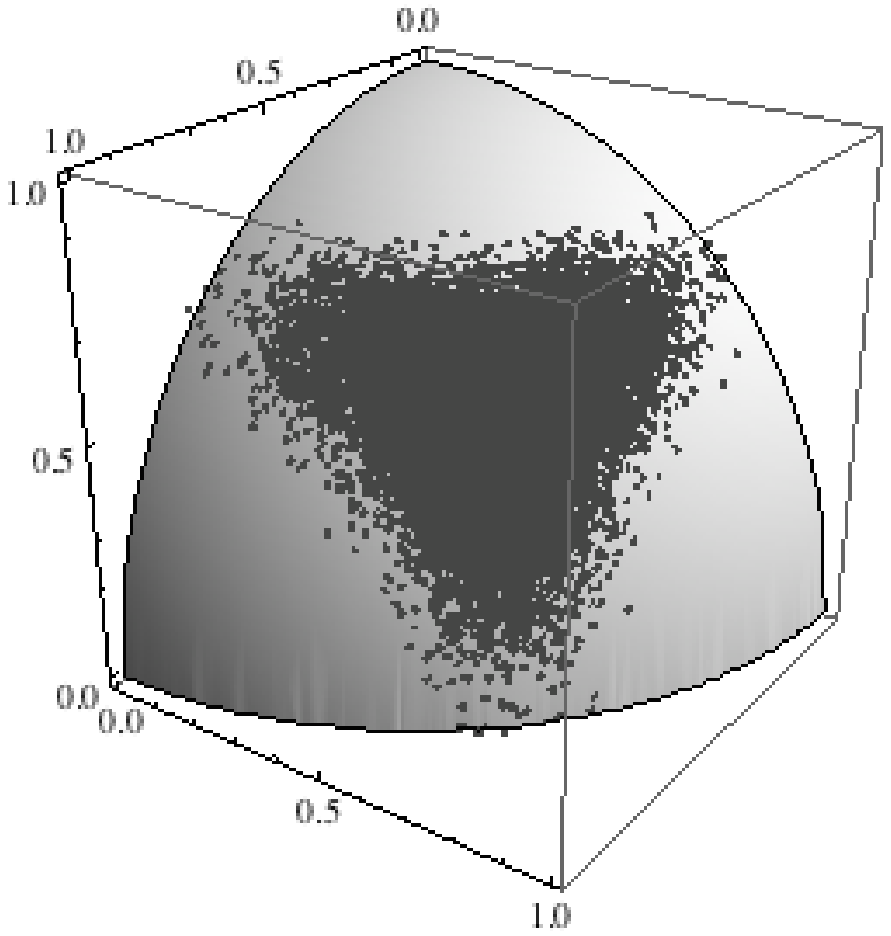}

 \caption{Simulated data (Model (1)-(A) with parameters estimated from \cite{Nuhn_Binder_Taylor_Halling_Hibbett:2013}}
 \label{fig:BolFake}
 ~ 
 \end{subfigure}
 \caption{A comparison of real versus simulated dissimilarity maps.  Data points projected onto the unit sphere.}
\end{figure}

\begin{figure}[ht]
 \centering
 \begin{subfigure}{0.5\textwidth}
 \includegraphics[width=0.9\linewidth]{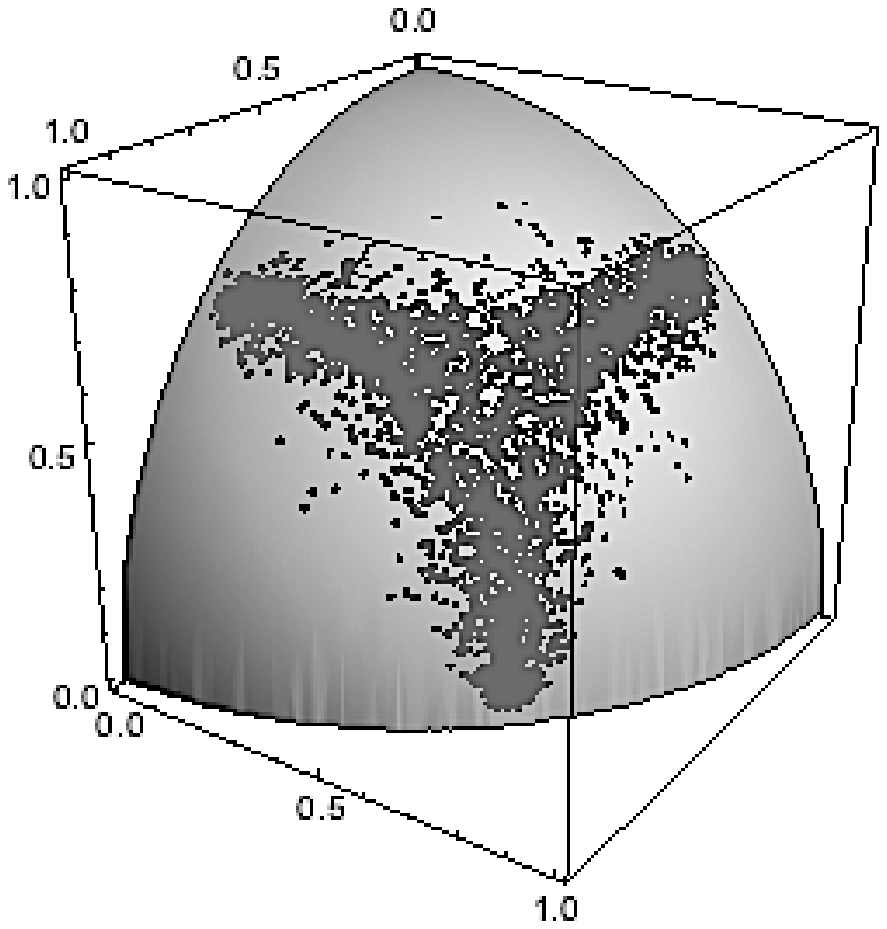}

 \caption{Data points from \cite{Spooner_Rojas_Bonierbale_Mueller_Srivastav_Senalik_Simon:2013}}
 \label{fig:DacReal}
 \end{subfigure}
 ~ 
 \begin{subfigure}{0.5\textwidth}
 \includegraphics[width=0.9\linewidth]{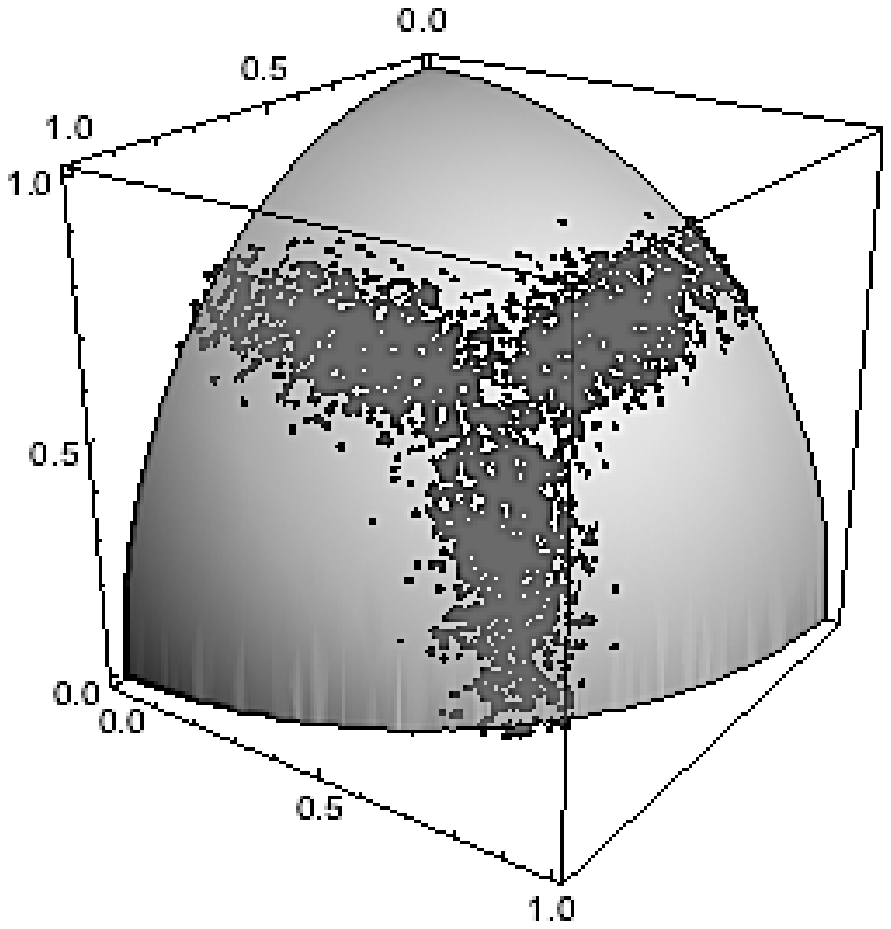}
 \caption{Simulated data (Model (1)-(A) with parameters estimated from \cite{Spooner_Rojas_Bonierbale_Mueller_Srivastav_Senalik_Simon:2013}}
 \label{fig:DacFake}
 ~ 
 \end{subfigure}
 \caption{A comparison of real versus simulated dissimilarity maps.  Data points projected onto the unit sphere.}
\end{figure}

\begin{figure}[h!]
\centering

\includegraphics[width=0.6\textwidth]{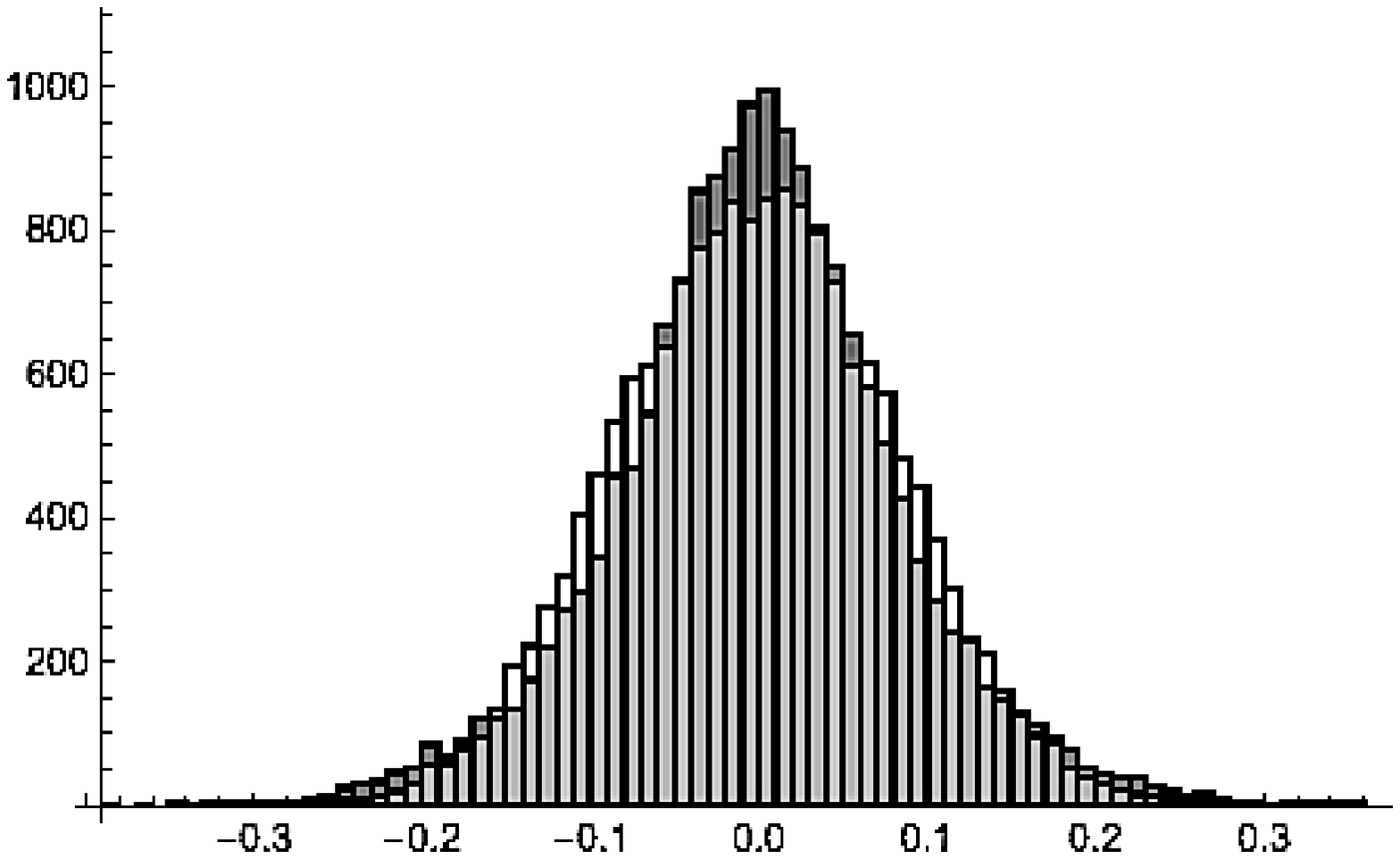}

\caption{Histograms of Real and Simulated Data from Model (1)-(A) $D_{1}| D_{2}, D_{3} $ from \cite{Nuhn_Binder_Taylor_Halling_Hibbett:2013}}
\label{fig:HistD1}
\end{figure}

\begin{figure}[h!]
\centering

\includegraphics[width=0.6\textwidth]{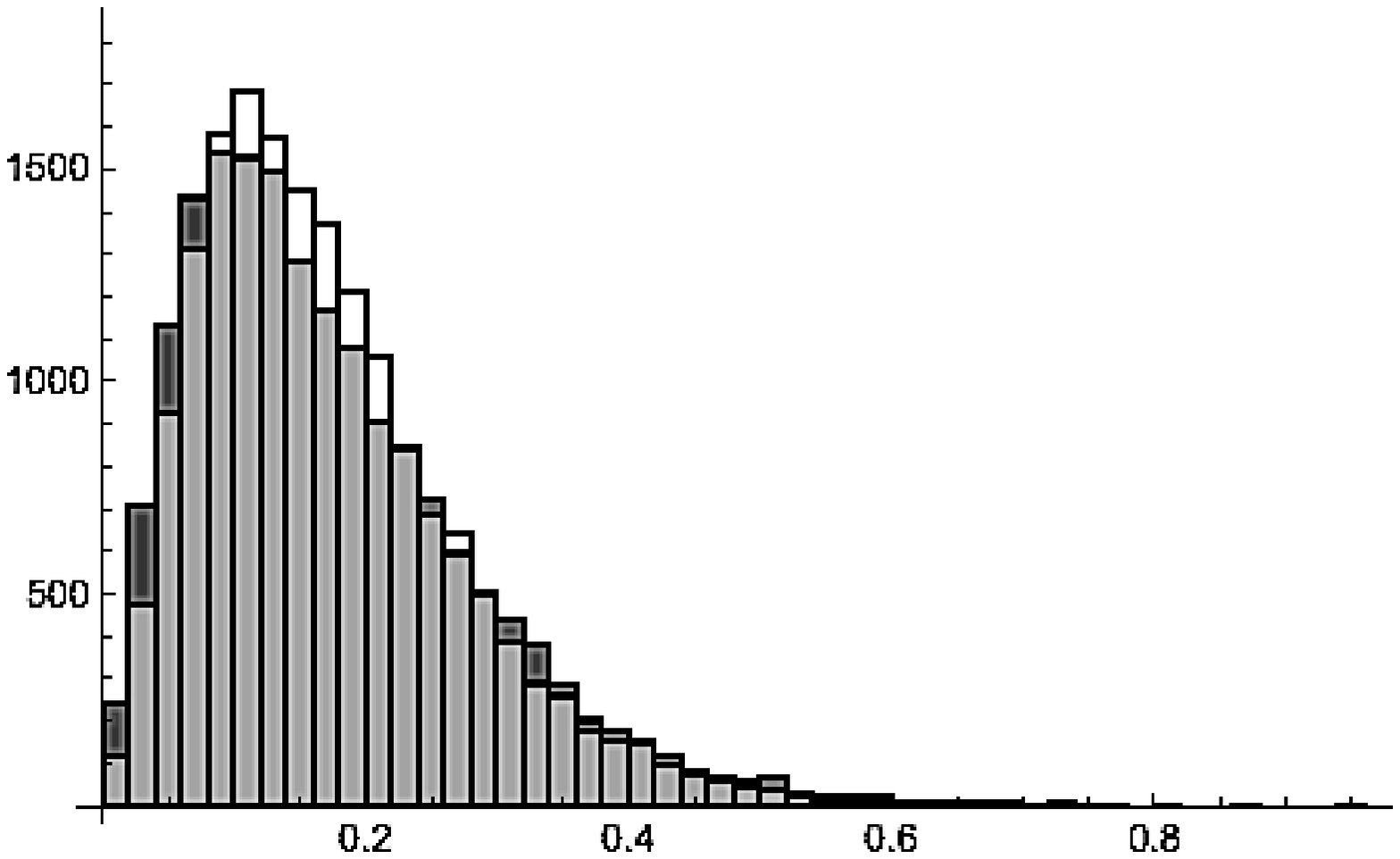}

\caption{Histograms of Real and Simulated Data from Model (1)-(A) $D_{2} | D_{3} $ from \cite{Nuhn_Binder_Taylor_Halling_Hibbett:2013}}
\label{fig:HistD2}
\end{figure}

\begin{figure}[h!]
\centering

\includegraphics[width=0.6\textwidth]{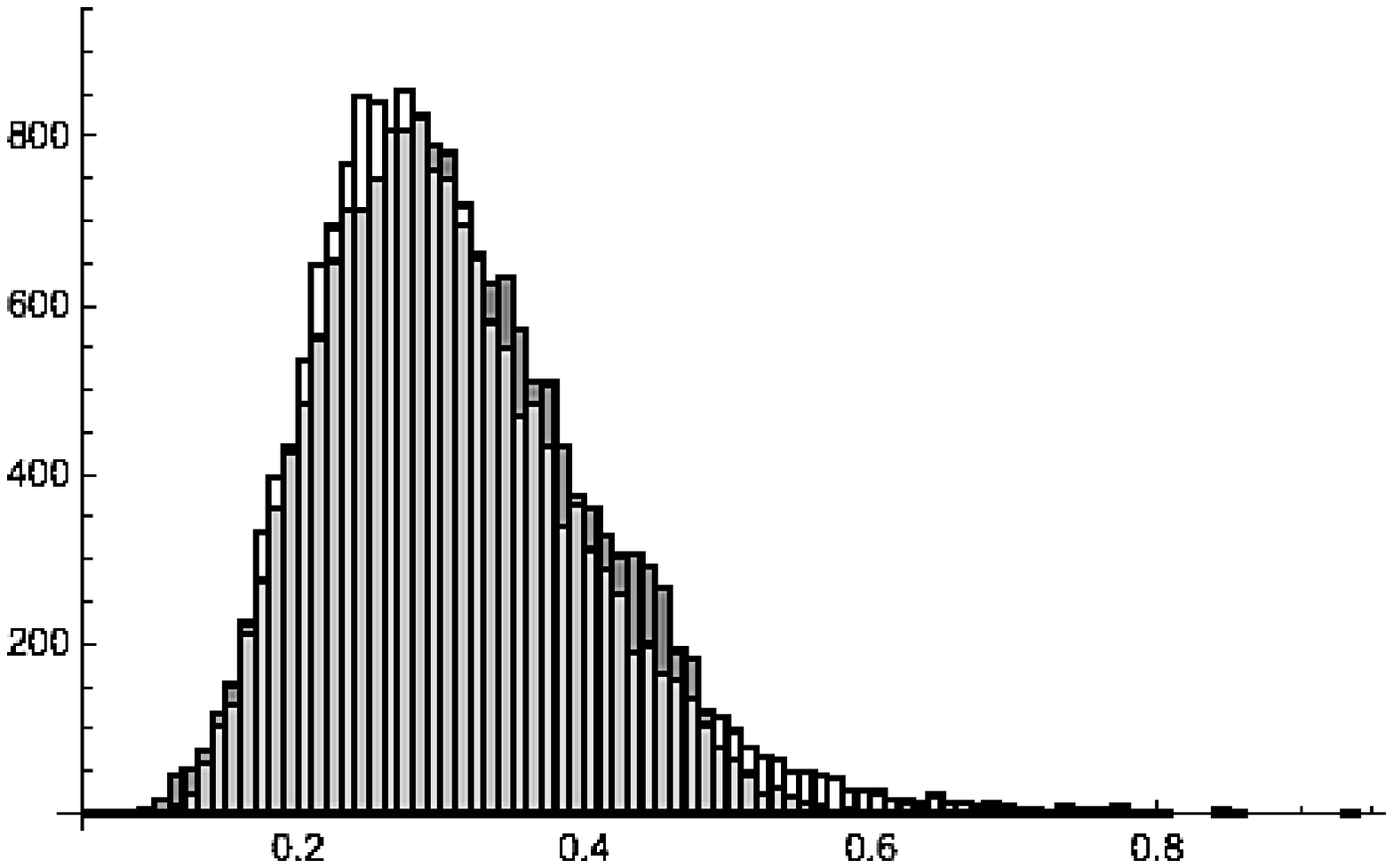}

\caption{Histograms of Real and Simulated Data from Model (1)-(A) $ D_{3} $ from \cite{Nuhn_Binder_Taylor_Halling_Hibbett:2013}}
\label{fig:HistD3}
\end{figure}

\section{Model Family (2)}\label{sec:YH}

To describe Model Family (2), we first describe the Yule-Harding pure birth process (Algorithm \ref{yule_alg}) which generates random rooted, binary phylogenetic trees with leaf set $[n]$. This is also commonly known as the Yule Process but we do emphasize that extinction is often incorporated into the model, in which case it is called a ``birth-death" process. We ignore the ``death" part of the model for simplicity. Recall that a \emph{pendant edge} is an edge in a phylogenetic tree incident to a leaf.

\begin{algorithm}[h!]
\caption{Yule Process}
\label{yule_alg}

\begin{itemize}

\item Input: Leaf set $[n] = \{1, \dots , n \}$, $n \geq 2$. 

\item Output: a rooted, binary phylogenetic $[n]$-tree. 

\item Initialize: Randomly select two leaves $x $ and $y$ from $[n]$ with uniform probability. Identify these two leaves as the leaf set of a rooted binary tree $T_{1}$. Set $S_{1} = [n] \setminus \{x,y \}$ and $L_{1} = \{ x, y \}$. 

\item While $S_{i} \neq \emptyset$:

\begin{itemize}

\item Randomly select an element $x_{i}$ of $S_{i}$ with uniform probability. 

\item Randomly select a pendant edge $e_{i} = (u_{i}, y_{i} ) $ of $T_{i}$ with probability determined by the uniform distribution on the set of pendant edges of $T_{i}$. Here $u_{i}$ is a binary internal vertex and $y_{i} \in L_{i}$. 

\item Subdivide $e_{i}$ by adding a new vertex $v_{i}$. 

\item Update $T_{i + 1}$ as the tree with leaf set $L_{i + 1} = L_{i} \cup \{ x_{i} \} $ , $V(T_{i + 1}) = V(T_{i}) \cup \{ v_{i}, x_{i} \}$ and $E(T_{i + 1}) = E(T_{i}) \setminus \{ e_{i} \} \cup \{ (u_{i}, v_{i}), (v_{i}, y_{i}), (v_{i}, x_{i}) \}$. Update $S_{i + 1} = S_{i} \setminus \{ x_{i} \}$. 
\end{itemize}

\item Return $T_{n -2 + 1}$, a rooted, binary phylogenetic tree with leaf set $[n]$.

\end{itemize}

\end{algorithm}

Given a tree $T$ with $n$ leaves generated under the Yule-Harding pure birth model we investigate the distributions for $D_2$ and $D_3$ for a subtree obtained by restricting $T$ to three randomly selected taxa. We define $D_{1}$ in the same way as for Model Family (1). We determine in this section that $D_2$ and $D_3$ should follow Gamma distributions. Throughout this section we will refer to the \emph{graph} distance in a tree $T$, which is the number of edges in the unique path between a pair of taxa in $T$. 

As mentioned above, dissimilarity maps estimated from DNA sequences is almost never ultrametric, but we use ultrametric trees - i.e. the geometric object $\mathcal{ET}_{3}$ - as we did for Model Family (1) as a reference for constructing Model Family (2). We again consider the perpendicular projection onto $\mathcal{ET}_{3}$ of a dissimilarity map $\delta$, so that $\mathcal{ET}_{3}$ is a baseline object for constructing our model. We consider three cases, which are represented by three labeled tree topologies pictured in Figure \ref{fig:threetrees}. We may again, as in Model Family (1) assume that it is equally likely that a data point $\delta$ is closest to one of the three wings. 

The three topologies are illustrated in Figure \ref{fig:threetrees} with branch lengths labeled, where $a$ and $b$ separate the two most closely related taxa, $c$ separates the single internal node of the tree and the root, and $d$ goes from the root to the third taxa. 
\[\begin{array}{ccc}
\Wing 1& =& (b+c+d,\frac{2a+b+c+d}{2},\frac{2a+b+c+d}{2}) \\
\Wing 2& =& (\frac{a+2b+c+d}{2},a+c+d,\frac{a+2b+c+d}{2})\\
\Wing 3& =& (\frac{a+b+2c+2d}{2},\frac{a+b+2c+2d}{2},a+b)\end{array}\]
Under the assumption that $T$ is ultrametric we can substitute $A=a=b$, $B=c=d-a$. Since these trees are subtrees of an $n$-taxon tree it is important to note that $A$ and $B$ are likely the sum of multiple branch lengths in the larger tree $T$. We can use the projections and corresponding ultrametric coordinates to give $A$ and $B$ as functions of the input data. 
\[\begin{array}{cll}
\Wing 1: & A=\frac{b+c+d}{2} & B=\frac{2a-b-c-d}{4}\\
\Wing 2: & A=\frac{a+c+d}{2} & B=\frac{2b-a-c-d}{4}\\
\Wing 3: & A=\frac{a+b}{2} & B=\frac{2c+2d-a-b}{4}
\end{array}\]

\begin{figure}
\centering
\includegraphics[scale=0.5]{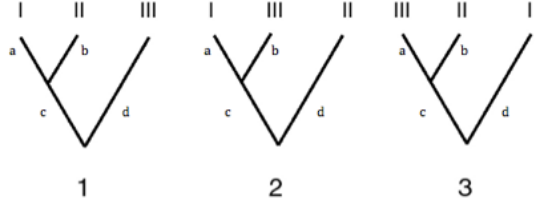}
\caption{Three taxa tree topologies with edge lengths.}
\label{fig:threetrees}
\end{figure}

\subsection{Spindle-origin distance}

\begin{prop}
Under the Yule-Harding model, we may describe the distribution of the set of unscaled distances of the projection from wing to spindle, $D_2$, as follows: $D_2 | \alpha_{D_2}$ has a Gamma distribution with parameters $\frac{2}{\sqrt{3}}\lambda$ and $\alpha_{D_2}$, where $\lambda$ is the birth rate and $\alpha_{D_2}$ has expected value $E[\alpha_{D_2}]=3(\mu_n-2(1-\frac{\mu_n}{n-1}))$ where $\mu_n=\sum\limits_{j=2}^n\frac{1}{j}\approx ln(n)+\gamma+\frac{1}{2n}-\frac{1}{12n^2}$, and $\gamma\approx0.5772156649$ denotes the Euler-Mascheroni constant. 
\end{prop}

\begin{proof}

The sum of $n$ independent exponential random variables with parameter $\lambda$ is a Gamma distribution with parameters $\alpha=n$ and $\lambda$ \cite{LeemisMcQueston}. We use this to show $D_2 | \alpha_{D_2}$ follows a Gamma distribution, and then determine the expected value of $\alpha_{D_2}$ based on the graph distances between the three taxa. 

Suppose we fix a three-taxon subtree $T_{0}$. We compute $D_2$ in terms of the branch lengths of $T_{0}$ by finding the Euclidean norm of the orthogonal projection from the point on the wing to the spindle. We find that the $D_2$ is given by
\[\frac{2}{\sqrt{3}}D,\]
where $D=2A+2(A+B)=a+b+c+d$. We note that $D$ is the sum of the branch lengths of $T_{0}$. Let the number of branches included in $D$ be $\alpha_{D_2}.$ Under the Yule-Harding model branch lengths follow an exponential distribution with parameter $\lambda$. Thus, $D$ is the scalar multiple of the sum of exponential branch lengths, which implies $D$ has distribution $\Gamma(\lambda, \alpha_{D_2})$ and $D_2$ follows a Gamma distribution under the Yule-Harding pure birth model. Notice that here $\alpha_{D_2}$ is a constant given the fixed subtree.

Then we see that $D=a+b+c+d = \frac{1}{2}(d(\mbox{I}, \mbox{II}) + d(\mbox{I}, \mbox{III}) + d(\mbox{II}, \mbox{III}))$, which implies $E[D] = \frac{1}{2}(E[d(\mbox{I}, \mbox{II}) ]+ E[d(\mbox{I}, \mbox{III})] + E[d(\mbox{II}, \mbox{III})])$. Since the subtree is randomly chosen, we have $E[D] = \frac{3}{2} E[d(\mbox{I}, \mbox{II})]$. Therefore $E[D] = E[E[D|\alpha_{D_2}]] = E[\alpha_{D_2}/\lambda] = E[\alpha_{D_2}]/\lambda = \frac{3}{2} E[d_{ij}]/\lambda$, i.e. $E[\alpha_{D_2}]= \frac{3}{2} E[d_{ij}]$, where $E[d_{ij}]$ is the expected graph distance between a random subset of two taxa in an Yule tree with $n$. Steel and McKenzie \cite{SteelMcKenzie} showed 
\[E[d_{ij}]=2\left(\mu_n-2\left(1-\frac{\mu_n}{n-1}\right)\right).\] 
Therefore \[E[\alpha_{D_2}] = 3\left(\mu_n-2\left(1-\frac{\mu_n}{n-1}\right)\right).\]

Additionally we note that the scalar multiple $kX$ of a random variable $X$ that follows a Gamma distribution with parameters $\alpha_{D_2}$ and $\lambda$ is also a Gamma distribution with parameters $\alpha_{D_2}$ and $k\lambda$. Thus we expect the spindle-to-origin distance to be $\Gamma\left(\alpha_{D_2},\frac{2}{\sqrt{3}}\lambda\right)$. 
\end{proof}

\subsection{Wing-spindle distance}

 \begin{prop}
Under the Yule-Harding model, the distribution of the set of distances of the projection from wing to spindle can be described as follows: $D_3 | \alpha_{D_3}$ has a Gamma distribution with parameters $2\sqrt{\frac{2}{3}}\lambda$ and $\alpha_{D_3}$, where $\lambda$ is the birth rate. 
\end{prop} 

We use similar argument to show $D_3 | \alpha D_{3}$ also follows a Gamma distribution. The Euclidean norm of the projection from the spindle to the origin is $2\sqrt{\frac{2}{3}}B$. For a fixed subtree under the Yule-Harding pure birth model, $B$ is the sum of exponential random variables and therefore follows a Gamma distribution. 

In order to determine the expected value for the shape parameter we need to know the expected graph distance in $T$ between the unique internal node and the root in $T_{0}$. Steel and McKenzie \cite{SteelMcKenzie} gave the expected value for the distance from the unique internal node of two taxa to the root in $T_{0}$. It is possible that a refinement of this work could lead to an explicit formula for the expected value of the shape parameter in terms of the number of taxa. But without such an explicit formula it is simpler to estimate the parameter directly from the data.

\section{Data analysis}
\subsection{Method for scoring the deviation of a model distribution}\label{sec:score}
Given a distribution function for pairwise distance data, one might wish to measure how well it conforms with the data it was designed to model.

In this section we present a method for scoring a model distribution's deviation from a data sample on $HS_{\ge 0}$. We use this method to compare how well our model families, as well as the uniform distribution, model a test data set. We use a numerical test statistic motivated by the assumption that for large samples, the proportion of the sample in any region of $HS_{\ge 0}$ should be the same as the integral of the proposed density function over that region.

Let $S$ be a sample of $N$ triples of pairwise distance data which has been orthogonally projected onto $HS_{\ge 0}$. Let $P$ be a random sample of triples of distance data drawn from a distribution function $\gamma$ and then projected onto the unit sphere. We define a test statistic to determine how different $S$ and $P$ are.

We select $k$ randomly generated points $v_i$ on the $HS_{\ge 0}$. Each such point divides $HS_{\ge 0}$, into three spherical triangles $$ T_1 =\{v_i,(1,0,0),(0,1,0) \}, $$ $$ T_2 =\{v_i,(1,0,0),(0,0,1) \}, \quad \text{and} \quad T_3 =\{v_i,(0,1,0),(0,0,1) \}.$$

For each $1 \le i \le k$ we define the triples $\hat{S}$ (resp $\hat{P}$) = $(s_1,s_2,s_3)$ where ${s_j= \frac{|S \cap T_j|}{|S|}}$. Then we can define the test statistic:
\begin{equation} \rho(N,k) = \frac{\displaystyle\sum_{i=1}^{k} \sqrt{(s_1-p_1)^2+(s_2-p_2)^2+(s_3-p_3)^2}}{k}. \end{equation} For large values of $N$ and $k$, $\rho$ should approach zero.

Note that if a density function $f$ was known for $\gamma$, then the $p_i$ could be computed by integrating $f$ along $T_1,T_2$, and $T_3$. If $f$ is the density function for the uniform distribution, then as $n \to \infty$, $\hat{P}$ will approach \[ \left(\frac{Area(T_1)}{\frac{\pi}{6}} ,\frac{Area(T_2)}{\frac{\pi}{6}} ,\frac{Area(T_3)}{\frac{\pi}{6}} \right).\]

\subsection{Fitting data sets from Treebase.org to Model Families (1) and (2)}

We used custom Mathematica software available on our supplementary materials website to test the ability of our different model families to fit biological data sets. The results for ten data sets are shown in Table \ref{table:modelscores}. Rows are indexed by the ten data sets and the entries in the tables are the score based on our geometric test statistic explained in Section \ref{sec:score}.

We explain the labels for the columns: Models (1)-(A) and (1)-(B) are from Model Family (1) where (1)-(A) indicates the choice of the extreme value distribution and (1)-(B) indicates the choice of the Gamma distribution.  Uniform is the model where input points are uniformly distributed in the input space. Lower scores are best, as explained in Section \ref{sec:score}, and the best score for each data set is in bold text. Rows are indexed by the data matrix name on Treebase.org: M18755 is the data matrix we used from \cite{Nuhn_Binder_Taylor_Halling_Hibbett:2013}, and we were no longer able to find the study \cite{Spooner_Rojas_Bonierbale_Mueller_Srivastav_Senalik_Simon:2013} in Treebase.org by the time of this publication, so we have made the distances we calculated from that data matrix available in our online supplementary materials. The row corresponding to that data matrix is labeled ``Daucus."

In Table \ref{table:modelscores}, for each data set, we used $k = 10000$ sample points for the area subdivisions in our scoring system. For small data matrices, we used all possible triples of input points $(\delta(u,v), \delta(u,w), \delta(v,w))= (x,y,z)$, i.e. $N = {n \choose 3}$ where $n$ was the number of taxa in the matrix. For data matrices with 50 or more taxa, we used approximately $N = 20,000$ data points, where $N$ varied slightly because our method for random sampling of points $(\delta(u,v), \delta(u,w), \delta(v,w))$ produced duplicates. 

While the results in Table \ref{table:modelscores} are interesting, and definitely show that for these datasets, the uniform distribution perform poorly, analysis of many more data sets would be necessary to establish a clear trend.  It also appears that Model (1)-(B) which blends the geometric features of Model Family (1), with the choice of Gamma distribution motivated by Model Family (2), may provide the best overall accuracy.

\begin{table}[!ht]
\begin{center}
 \begin{tabular}{|c|c|c|c|c|}
 \hline
 Data Matrix & Model (1)-(A) & Model (1)-(B) &  Model 2 & Uniform \\
 \hline
 M1807 & 0.0208444 & 0.0276282  & \textbf{0.01446} & 0.0892374 \\
 \hline
 M1789 & \textbf{0.00944452} & 0.00961268  & 0.0154542 & 0.26513 \\
 \hline
 M2566 & 0.0624382 & \textbf{0.04826421} &  0.0736874 & 0.060538 \\
 \hline
 M25733 & 0.0140685 & \textbf{0.0131326} &  0.0334055 & 0.264626 \\
 \hline
 M18755 & 0.00826286 & \textbf{0.00819624} &  0.0134024 & 0.293712 \\ 
 \hline
 M25665 & 0.0141961 & \textbf{0.0141181}  & 0.0276884 & 0.237157 \\ 
 \hline
 Daucus & 0.00879329 & \textbf{0.00672433} & 0.0148475 & 0.294249 \\ 
 \hline 
 M536 & 0.0251412 & \textbf{0.019827} & 0.0215093 & 0.130195 \\
 \hline 
 M806 & \textbf{0.0282066} & 0.0331788 & 0.0359551 & 0.153135 \\ 
 \hline
 M1169 & \textbf{0.0115309} & 0.0122847  & 0.0177641 & 0.328366 \\ 
 \hline
 \end{tabular}
 \caption{Scores $\rho(N,k)$ from each model for data sets with Jukes-Cantor Distances from Treebase.org. Model of best fit is in bold.}
\label{table:modelscores}
\end{center}
\end{table}

\section{Motivations: choice of distributions in Model Families (1) and (2)}\label{sec:distributions}

The variance of the normal distribution parametrizes $D_{1}$ in both model families (1) and (2). A value of $d_{1} = 0$ means that a data point is in $\mathcal{ET}_{3}$, and so is an ultrametric additive distance matrix. Therefore larger $d_{1}$ should correspond to higher variations in the rates of evolution between the three taxa for the data point; i.e. deviation from the molecular clock. The normal distribution is thus a natural choice for both model families. 

We first chose the EVD to model $D_{2}$ and $D_{3}$ because histograms of the corresponding steps in path traces of the real data points closely resembled the probability density function for the EVD. Figures \ref{fig:HistD1}, \ref{fig:HistD2}, and \ref{fig:HistD3} show the histograms of data points of the type $( \delta(u,v), \delta(u,w), \delta(v,w) )$ plotted together with histograms of data points generated using Model (1)-(A). However, since both the Gamma and EVD are transformations of the exponential distribution, \cite{LeemisMcQueston}, and because of our results from Section \ref{sec:YH}, we felt it was necessary to include Model (1)-(B). 

Our model families seem to be capable of capturing rather subtle features of the biological data beyond how far the data deviate from the molecular clock. We observed that in some datasets, such as in \cite{Spooner_Rojas_Bonierbale_Mueller_Srivastav_Senalik_Simon:2013}, data points are less distributed in the center (i.e. around the spindle). This phenomenon is visible in Figures \ref{fig:DacReal} and \ref{fig:DacFake}. This means that the distribution of $D_2$ has lower density when $D_2$ is close to 0 and near the boundary. We comment on this feature for this data set because the density functions of the EVD and Gamma distributions can accommodate this feature of the data with appropriate parameter choices in both Model Families (1)-(A) and (1)-(B), while the density function for the exponential distribution on branch lengths used in Model Family (2) cannot. 

\section{Non-ultrametric trees and application to larger data sets}
The model families described below emphasize the geometry of ultrametric trees ($\mathcal{ET}_3$).  However, a similar analysis could be conducted relative to the more general tree space $\mathcal{T}_{n}$.  With only three taxa, every dissimilarity map that satisfies the triangle inequality is already realizable as a pairwise distance on a tree \cite{SempleSteel}.  Therefore, it is more interesting to consider the relationship between  maps and the more general tree space in the context of $\mathcal{T}_4$. 

The geometry of $\mathcal{T}_{4}$  can be described as the union of three \emph{superwings} $$W_{14}=\{P\in \mathbb{R}^6|x_12+x_34=x_{13}+x_{24} \le x_{14}+x_{23} \},$$ $$W_{13}=\{P\in \mathbb{R}^6|x_12+x_34=x_{14}+x_{23} \le x_{13}+x_{24} \},$$ and $$W_{12}=\{P\in \mathbb{R}^6|x_13+x_24=x_{14}+x_{23} \le x_{12}+x_{34} \}.$$ While this model begins with an emphasis on general tree space, we shift our emphasis to the condition of ultrametricity through a projection to $\mathcal{ET}_{3}$ at a later stage. 

 In the earlier sections of this chapter we assumed our biological data arose from the addition of Gaussian noise to phylogenetic trees identified as points in Euclidean space.  We continue with this paradigm and model the distance of a dissimilarity map to the nearest superwing using a normal distribution.  Next we project the image of the point on the superwing to the \emph{superspindle} or set of coordinates with $x_{12}+x_{34}=x_{13}+x_{24}=x_{14}+x_{23}$. This distance could be modeled by a Gamma distribution.

As the sum of branch lengths under the Yule-Harding model we can assume $x_{12}+x_{34}$ follow a Gamma distribution. By knowing this sum we can reconstruct all coordinates using only $x_{12},x_{13}$, and $x_{23}$. These triples of coordinates lie in $\mathbb{R}^3$ and correspond to triples of pairwise distances drawn from a tree on $n$ taxa and can thus be modeled using the findings in Section~\protect\ref{sec:twomodels}.

Moreover, when there are more than four taxa the geometry of tree space is simply the intersection of superwings over all sets of four taxa. This follows directly from Buneman's four-point condition (or three point condition in the ultrametric setting) \protect\cite{buneman1974note}.  While it is possible to build more complicated distributions relating dissimilarity maps to $\mathcal{T}_{n}$ or $\mathcal{ET}_{n}$ such models would require a larger set of parameters, and may necessarily contain more information than the minimal cases.  However, the following question is paramount for turning the theoretical work of this chapter into a practical tool for computational biologists.
\begin{ques} How does a distribution of dissimilarity maps on subsets of $3$ (or $4$) of $n$ taxa induce a distribution of dissimilarity maps on $n$ taxa?
\end{ques}

\section{Applications and future work}\label{sec:future}
The models in this chapter provide new tools for interpreting analytic results about distance-based phylogenetic methods using geometry. Analyses of trends in the shapes of trees using notions such as balance statistics has led to interesting insights and questions about model assumptions in the past \cite{Aldous}.  The tools presented in Sections \ref{sec:model1} and \ref{sec:YH} establish a framework for evaluating suitability of model assumptions for phylogenetic data. Recently, the TreeBASE website (accessed June 15, 2016) stated: ``as of April 2014, TreeBASE contained data for 4,076 publications written by 8,777 different authors" \cite{Treebase}. A systematic analysis of more alignment matrices using the software available in our supplementary materials might give a clearer picture of trends in biological data that relate the model families presented here, providing empirical support for choosing one model family as more suitable for further development and study. 

Our use of the geometric object $\mathcal{ET}_{3}$ as the backbone of our models serves many purposes. For example, constructing Model Families (1) and (2) around $\mathcal{ET}_{n}$ provides a unifying framework to study which Model Family is better suited to modeling biological data.  The Yule-Harding pure birth process does produce ultrametric trees, which are points in $\mathcal{ET}_{n}$.  So, comparing the Model Families (1) and (2) in this geometric setting may provide insight into the implications of assuming the Yule-Harding model in data simulation.  

Furthermore, the Mathematica software in our supplementary materials can be used directly to estimate the deviation of a biological data set from the molecular clock, via the amount of variance in the normal distribution parameter estimated as $D_{1}$.  We have already called attention to the fact that biological data is rarely ultrametric, but in the biology and computer science communities, especially in the field of phylogenomics, the assumption of a molecular clock is still necessary to provide theoretical guarantees in many instances. For example, in \cite{RochWarnow} it is shown that without the assumption of a molecular clock, no theoretical guarantees bounding gene tree estimation error can be made. 

Also, some phylogenomic methods that bypass the need for gene tree estimation rely on the assumption of a molecular clock for their theoretical guarantees. For example, the recently developed method SVDquartets \cite{SVDquartets} relies on the identifiability result in \cite{ChifmanKubatko14} for theoretical guarantees, and requires a molecular clock. So, we argue that establishing the deviation of a data set from the molecular clock using our software would be useful in determining whether or not it is reasonable to apply a method that is only guaranteed to work well on data that follows a molecular clock to that specific data set. 

The results shown in this chapter are restricted to distances computed under the Jukes-Cantor model of sequence evolution. We performed the same experiments with the K2P model, but the results on the data sets we investigated were so similar for each data set that we omitted these findings. However, future work should include an analysis of both distances computed from amino acid data as well as more general statistical models of sequence evolution such as GTR \cite{GTR} and GTR+$\Gamma$ \cite{GTRGAMMA}. 

\subsection{Future Research Questions}

\begin{ques} Is there a correlation between the number of missing taxa in a tree and the parameters of the Gamma distribution in Model Family (2)? How does this shed light on how closely a set of sampled species may be related in the Tree of Life? 
\end{ques}

\begin{ques} What is the likelihood of the rogue taxa phenomenon using Model Family (2), which more naturally extends to $n > 3$ than Model Family (1)? \end{ques}

\begin{ques} Can one provide an inductive and computationally tractable model for the distribution of distance data for pairwise dissimilarity vectors on $n$ taxa for Model Family (1)? If so, how will this model compare to Model Family (2) \end{ques}

\begin{ques} For sequence data generated under the coalescent model, is there a correlation between the estimated Gamma distribution parameters across the different gene trees? Could the variation among these parameters inform the design of summary methods? \end{ques}

\section{Supplementary material}\label{sec:supp}
All supplementary materials, including the distance matrixes computed from the alignment data matrices downloaded from Treebase.org, and Mathematica software for generating simulated distance data, can be found at goo.gl/08PUC5.

\section{Acknowledgements}
R.D. was partially supported by the National Science Foundation (DMS 0954865). J.R. was partially supported by grants from the National Center for Research Resources (5 P20 RR016461) and the National Institute of General Medical Sciences (8 P20 GM103499) from the National Institutes of Health. Z.V. was partially supported by the National Institute of General Medical Sciences (8 P20 GM103499) from the National Institutes of Health. J. X. was partially supported by the David and Lucille Packard Foundation. We also thank James Degnan, Megan Owen, Mike Steel, Seth Sullivant, Caroline Uhler, Tandy Warnow, and Ruriko Yoshida for helpful discussions over the course of this project.







\bibliographystyle{amsplain}

\bibliography{refs.bib}

\end{document}